\documentclass[11pt]{article}

\usepackage[utf8]{inputenc}
\usepackage{amsmath, amssymb, amsthm}
\newtheorem{theorem}{Theorem}
\newtheorem{proposition}{Proposition}

\usepackage{fullpage}
\usepackage{setspace}
\usepackage{natbib}

\usepackage{graphicx}
\everymath{\displaystyle}

\title{Stationary distribution of the volume at the best quote \\ in a Poisson order book model}
\author{Ioane Muni Toke
\\ ERIM, University of New Caledonia
\\ BP R4 98851 Noumea CEDEX, New Caledonia
\\ \tt{ioane.muni-toke@univ-nc.nc}
\\ \vspace{6pt}
\\ Applied Maths Laboratory, Chair of Quantitative Finance
\\ Ecole Centrale Paris 
\\ Grande Voie des Vignes, 92290 Châtenay-Malabry, France
\\ \tt{ioane.muni-toke@ecp.fr}
}
\date{\today}

\begin{document}
\maketitle
\begin{abstract}
In this paper, we develop a Markovian model that deals with the volume offered at the best quote of an electronic order book. The volume of the first limit is a stochastic process whose paths are periodically interrupted and reset to a new value, either by a new limit order submitted inside the spread or by a market order that removes the first limit. Using applied probability results on killing and resurrecting Markov processes, we derive the stationary distribution of the volume offered at the best quote. All proposed models are empirically fitted and compared, stressing the importance of the proposed mechanisms.
\end{abstract}

\section{Introduction}

The limit order book has recently been the subject of a growing interest among academics and practitioners studying financial markets.
This electronic structure centralizes all the orders submitted to a given market by all participants, finds matchable buy and sell orders, and therefore defines the price of the financial product exchanged. The fundamental question is thus to understand how the sequences of submitted orders -- the orders flows -- are translated into price dynamics.
\citet{Biais1995} and \citet{Bouchaud2002} are pioneer investigations on the empirical properties of the limit order book.
\citet{Smith2003} is a pioneer theoretical framework for the study of the continuous double auction that is used in limit order books.

The order book is a complex system. Basic mathematical models, such as \citet{Cont2010}, rely on simplifying assumptions: only three main types of orders are submitted (limit, market and cancellations, ignoring exchange-specific rules and specificities) ; all orders flows are Poisson processes ; all orders have the same unit size.
\citet{Abergel2013} shows that under appropriate assumptions, some limit theorems apply and such Markovian models lead to a diffusion equation for the price. \citet{MuniToke2014} shows that with similar Poisson models, the average shape of the order book is analytically computable, even when relaxing the unit-volume hypothesis.

Among the important quantities that describe a limit order book, the volume offered at the best quote (bid or ask), i.e. the total number of shares available at the first limit of the order book, is fundamental. 
One reason is that this quantity is often the only information easily available to market participants (known as Level-1 data).
Another (linked) reason is that this quantity is heavily used in trading strategies by market participants : \citet{Farmer2004} shows that the majority of market orders that move the price have a size exactly equal to the volume of the best quote at the time of submission.

However, the volume of the best quote remains difficult to grasp in a model, since market events make it vary widely. When the volume of the best quote drops to zero (because of large market orders or cancellations that remove all the liquidity of the best quote), or when a new best quote is submitted inside the spread, the price changes, and the volume of the best quote is reset to a new quantity that may have no link to the quantities describing the last order submitted. 
Therefore, computing the volume of the best quote after a new event requires to keep track of the whole order book.

Recently, \citet{Cont2012} proposes a model in which the order book is restricted to its first limits. When the price does not move, the volume at the best quote obviously varies according to the arriving orders flows, and when this volume drops to zero, the price moves and the volume at the best quote is immediately reset to some random value. Therefore, the volume of the best quotes (bid and ask) is a two-dimensional process with values on the positive orthant that jumps randomly inside the orthant each time it reaches an axis. \citet{Cont2012} shows that under appropriate assumptions and using limit theorems in the spirit of queueing theory, this volume may approximately exhibit a jump-diffusive behaviour.
There is however a very restrictive assumption for this model to be valid: the spread has to always be equal to one tick. Indeed, one cannot allow for limit orders submitted inside the spread in this framework, since it would make the process jumps even when it does not reach an axis. Such an assumption may be an appropriate model for some very busy periods of trading of very liquid stocks, but probably not in the general case.

In this paper, we show that in the classical zero-intelligence order book model with Poisson processes, the stationary distribution of the volume offered at the best quote can be analytically and then numerically computed. The proposed model is basic but flexible. It does not assume that the spread is always equal to one tick, i.e. all types of events that make the volume of the best quote jump are taken into account: aggressive market orders that matches the full first limit, as well as aggressive limit orders submitted inside the spread as well.
The main idea is that such jumps of the volume at the best quote can be identified as killing and resurrecting a Markov process \citep{Pakes1997}.
The assumptions that all orders have to be unit-sized can even be lifted, with additional restrictions on the volume of market orders.

The remainder of the paper is organized as follows. 
Section \ref{sec:GeneralModel} describes precisely the general zero-intelligence model of the order book with Poisson processes that is used here. The main result on killing and resurrecting Markov processes is recalled and adapted to the order book context.
Sections \ref{sec:ModelsWithAggressiveMarketOrdersExclusively} and \ref{sec:ModelsWithBothMarketOrders} then explore two types of restrictive assumptions that allow the analytical computation of the stationary distribution of the volume at the best quote: Section \ref{sec:ModelsWithAggressiveMarketOrdersExclusively} excludes market orders that partially match the best quote, while Section \ref{sec:ModelsWithBothMarketOrders} allows for all types of market orders but with some size restrictions.
Finally, Section \ref{sec:EmpiricalResults} provides empirical fittings of all the analytical models proposed, with a comparison to the simulated "best effort" of the unrestricted Poisson model.

\section{A Markovian model of the one-side order book}
\label{sec:GeneralModel}

Let us consider the best ask quote of an electronic order book (the model for the best bid is strictly identical). The quantity offered at the best quote may be modified by the arrival of market or limit orders and by cancellation of existing limit orders.

Firstly, we assume that limit orders at the best quote are submitted according to a Poisson process with rate $\lambda_1$ and that the size of limit orders at the best quote form a set of independent and identically distributed random variables with probability distribution $(g_{1,n})_{n\in\mathbb N}$.
Secondly, we also assume that each unit-size component of a limit order (i.e. each share) standing at the best quote is cancelled some random time after its submission. All these random times are assumed to form a set of independent and identically distributed random variables with exponential distribution with parameter $\theta_1>0$.
Thirdly, we assume that we observe "partial" market orders (the meaning of partial will soon be clear), that are submitted according to a Poisson process with rate $\mu$ and are all unit-sized (one share).
We also assume that the cancellation mechanism and the "partial" market orders cannot fully deplete the best quote. In other words, the last share of the best quote cannot be cancelled, and cannot be matched by a partial market order (hence the name of the "partial" market orders, which match only partially the best quote).
This small assumption ensures that classic cancellations and market orders cannot make the price move.

Two types of orders can make the price move: the aggressive market orders and aggressive limit orders. Orders of the first type, aggressive market orders, are submitted according to a Poisson process $\mu_A$, and their size is equal to the available quantity at the best quote at the time of the submission.
In other words, aggressive market orders match all the liquidity available at the best quote. This is not an unreasonable assumption: \cite{Farmer2004} shows that on a 16-stock sample from the London stock exchange, 86\% of the buy market orders that change the price have a size that is exactly equal to the volume offered at the best ask.
When an aggressive market order is submitted, the volume available at the best quote drops to zero, the price moves up and the quantity available at the best limit is instantaneously reset to the quantity available at the next non-empty limit of the order book. 
This next non-empty limit of the order book evolves according to the submission and cancellation of limit orders. We may already specify, for future use, a model that is coherent with the Markovian setting we are establishing. To this end we will assume the following: the limit orders arrive according to a Poisson process with rate $\lambda_2$ ; the sizes of these limit orders form a set of independent and identically distributed random variables with probability distribution $(g_{2,n})_{n\in\mathbb N}$ ; each unit-size component of a limit order (i.e. each share) standing at the next non-empty limit is cancelled some random time after its submission, and all these random times are assumed to form a set of independent and identically distributed random variables with exponential distribution with parameter $\theta_2>0$.

The second type of orders that move the price are the aggressive limit orders. These are limit orders that are submitted inside the spread, i.e. at a price lower than the current best ask.
We will assume that aggressive limit orders are submitted according to a Poisson process with rate $\lambda_0$, and that the size of all aggressive limit orders form a set of independent and identically distributed random variables with probability distribution $(g_{0,n})_{n\in\mathbb N}$.
The effect of the submission of an aggressive limit order is simple: at the moment of the submission, this order instantly becomes the best quote, i.e. the ask price is reset to the price of the aggressive limit order, and the quantity available at the best quote is reset to the volume of the submitted limit order.

In summary, the continuous-time stochastic process $X=\left\{X(t),t\in[0,\infty)\right\}$ describing the volume available at the best quote evolves as follows.
Let $\tau_1$ be the random time of the first price move. During the time interval $[0,\tau_1)$, in the absence of events that move the price (aggressive market orders, aggressive limit orders), $X$ evolves as the stochastic process $1+Y=\{1+Y(t),t\in[0,\infty)\}$, which is one (the last share that cannot be cancelled or executed) plus the size of a queue with the infinitesimal generator:
\begin{equation}
\label{eq:InfGenFL}
\left(\begin{array}{cccccc}
-\lambda_1 & \lambda_1 g_{1,1} & \lambda_1 g_{1,2} & \lambda_1 g_{1,3} & \lambda_1 g_{1,4} &\ldots
\\
\mu+\theta_1 & -(\mu+\lambda_1+\theta_1) & \lambda_1 g_{1,1} & \lambda_1 g_{1,2} & \lambda_1 g_{1,3} &\ldots
\\
0 & \mu+2\theta_1 & -(\mu+\lambda_1+2\theta_1) & \lambda_1 g_{1,1} & \lambda_1 g_{1,2} &\ldots
\\
0 & 0 & \mu+3\theta_1 & -(\mu+\lambda_1+3\theta_1) & \lambda_1 g_{1,1} &\ldots
\\
\vdots & \vdots & \ddots & \ddots & \ddots & \ddots 
\end{array}\right).
\end{equation}

Now, at time $\tau_1$, the price moves because of an aggressive limit or aggressive market order. The process $X$ is instantaneously reset to a new random variable $X_{\tau_1}=H$ that, depending on the direction of the price move, represents either the size of the incoming aggressive limit order (downward price move), or the volume offered at the next non-empty limit inside the order book (upward price move).
Then, if $\tau_2$ is the random time of the next price movement, $X$ on $[\tau_1,\tau_2)$ behaves according to the infinitesimal generator \eqref{eq:InfGenFL}, and so on.

This mechanism is identifiable to what is known in applied probability as killing and resurrecting a Markov process. 
The process of the volume of the best quote starts at time $0$ and evolves according to the infinitesimal generator \eqref{eq:InfGenFL}. Then, upon the submission of an aggressive limit or market order, it is killed, and (instantaneously) resurrected to a random variable $H$ with distribution $(h_i)_{i\in\mathbb N^*}$, from where it restarts its course according the previous dynamics. Such a mechanism is studied in \cite{Pakes1997}, where the following result is proved.

\begin{theorem}[rephrased from \cite{Pakes1997}]
Let $Z=\{Z_t,t\in[0,\infty)\}$ be a Markov process on $\mathbb N$ with initial distribution $(h_i)_{i\in\mathbb N}$. Zero state is assumed to be absorbing for $Z$. Let $R=\{R_t,t\in[0,\infty)\}$ be the process constructed as follows : $R$ starts following some path of the process $Z$ ; at some random time (exponentially distributed with parameter $\beta>0$), $R$ is killed, i.e. reset to $0$ ; after some random time (exponentially distributed with parameter $\alpha>0$), $R$ is resurrected, i.e. restarts following a new path of the process $Z$ ; and so on.

Then the process $R$ admits a stationary distribution $(\pi_j)_{j\in\mathbb N)}$ given by:
\begin{equation}
	\pi_0 = \frac{\beta}{\alpha+\beta-\alpha\beta \hat f_0(\beta)}, \;\;\; \pi_j = \alpha\pi_0 \hat f_j(\beta),
\end{equation}
where $\hat f_j$ is the Laplace transform of the series $\sum_{i\in\mathbb N} h_i p_{i,j}(t)$, in which $p_{i,j}(t)=\mathbf P(Z(t)=j\vert Z(0)=i)$ is the transient probability of the (non-killed) process $Z$ from state $i$ to state $j$.
\end{theorem}

In our order book model, the resurrection is instantaneous, i.e. $\alpha\to +\infty$ ; the killing events are aggressive limit and market orders, i.e. $\beta = \lambda_0+\mu_A$ by standard properties of the Poisson processes ; the state $0$ is not accessible without killing (classic cancellations and partial market orders cannot deplete the best limit), i.e. $\hat f_0=0$. 
Another consequence of the last fact is that the distribution $(h_i)_{i\in\mathbb N^*}$ of $H$ represents exactly the new volume available at the best quote after an aggressive event. This distribution is therefore a mix of the distribution $(g_{0,i})_{i\in\mathbb N^*}$ with probability $\frac{\lambda_0}{\mu_A+\lambda_0}$ and the stationary distribution of the volume at the second limit, denoted $(\pi_{2,i})_{i\in\mathbb N}$, with probability $\frac{\mu_A}{\mu_A+\lambda_0}$.
We thus obtain the following result.

\begin{proposition}
In the general Poisson order book model described in this section, the stationary distribution $(\pi_j)_{j\in\mathbb N^*}$ of the total volume offered at the best quote is written for any $j\geq 1$ :
\begin{equation}
	\label{eq:StatDistVolBestQuote}
	\pi_j = (\lambda_0+\mu_A) \hat f_j(\lambda_0+\mu_A)
\end{equation}
where $\hat f_j(\cdot)$ is the Laplace transform of the series $\sum_{i\in\mathbb N} h_i p_{i,j}(t)$
in which $\forall i\in\mathbb N^*$, 
\begin{equation}
	h_i = \frac{\lambda_0 g_{0,i}}{\lambda_0+\mu_A} + \frac{\mu_A\pi_{2,i}}{\lambda_0+\mu_A},
\end{equation}
and $p_{i,j}(t)=\mathbf P(1+Y(t)=j\vert 1+Y(0)=i)$ is the transient probability of the (non-killed) process $1+Y$ from state $i$ to state $j$.
\end{proposition}

This result is central in this work.
In the following, we will study several specifications of the above general model, all of them allowing analytical tractability at some cost.
The first criterion dividing the different types of models is the presence or absence of "partial" market models. In the first type of model (Type-1 models), we assume that all market orders are aggressive market orders. In other words $\mu=0$, and there are no "partial" market orders that match only partially the best limit.
In the second type of models (Type-2 Models), this restriction is lifted, i.e. market orders may or may not be aggressive ($\mu>0$), but some restrictions on the distributions of the volumes will be added.
These different types of models are studied in the next two sections.

\section{Models with aggressive market orders exclusively}
\label{sec:ModelsWithAggressiveMarketOrdersExclusively}

Type-1 models assume that all market orders are aggressive, i.e. that $\mu=0$.
In this setting, we are able to analytically compute the stationary distribution $\pi$ of the process $X$ through a direct approach: we compute the transient probabilities $(p_{i,j}(t))_{i,j\in\mathbb N}$ by a standard generating function method, as well as the stationary distribution $(\pi_{2,i})_{i\in\mathbb N^*}$ by a direct method.

Let us start with the transient probabilities $(p_{i,j}(t))_{i,j\in\mathbb N}$. Let $p_{i,j}=r_{i-1,j-1}$ for any $(i,j)\in (\mathbb N^*)^2$. $(r_{i,j}(t))_{i,j\in\mathbb N}$ are the transition probabilities of the process $Y$.
The Kolmogorov forward equations are written for any $(i,j)\in\mathbb N^2$:
\begin{equation}
	r'_{ij}(t) = -(\lambda_1+j\theta_1) r_{ij}(t) + (j+1)\theta_1 r_{i,j+1}(t)+ \sum_{k=0}^{j-1} \lambda_1 g_{1,j-k}r_{i,k}(t).
\end{equation}
Let $G_1(z)=\sum_{j=1}^{+\infty} g_{1,j} z^j$ be the generating function of the distribution of the sizes of incoming limit orders at the best limit.
By multiplicating this equation by $z^j$ and summing over $j$, we obtain that the generating function $\varphi_i(z,t)=\sum_{j=0}^\infty r_{ij}(t)z^j$ is solution of the partial differential equation:
\begin{equation}
\label{eq:PDEvarphi}
0=\frac{\partial\varphi_i}{\partial t}(z,t)-\theta_1(1-z)\frac{\partial\varphi_i}{\partial z}(z,t)+\lambda_1(1-G_1(z))\varphi_i(z,t),
\end{equation}
subject to the initial condition $\varphi_i(z,0)=z^i$. 

As for the stationary distribution $(\pi_{2,i})_{i\in\mathbb N^*}$, we use the Markovian setting described in section \ref{sec:GeneralModel} : $\lambda_2>0$ is the rate of arrival of limit orders, $(g_{2,n})_{n\in\mathbb N}$ is the distribution of their sizes, and $1/\theta_2$ is the average lifetime of a share standing inside the book.
Similarly to what has been assumed for the best quote, we assume that the last share cannot be cancelled, so that the size of the queue does not drop to zero (it is by definition the next-non empty limit of the order book).
With these assumptions, the size of the book at the second limit is the process $\{1+Y_2(t),t\in[0,\infty)\}$  with infinitesimal generator
\begin{equation}
\label{eq:InfGenSL}
\left(\begin{array}{cccccc}
-\lambda_2 & \lambda_2 g_{2,1} & \lambda_2 g_{2,2} & \lambda_2 g_{2,3} & \lambda_2 g_{2,4} &\ldots
\\
\theta_2 & -(\lambda_2+\theta_2) & \lambda_2 g_{2,1} & \lambda_2 g_{2,2} & \lambda_2 g_{2,3} &\ldots
\\
0 & 2\theta_2 & -(\lambda_2+2\theta_2) & \lambda_2 g_{2,1} & \lambda_2 g_{2,2} &\ldots
\\
0 & 0 & 3\theta_2 & -(\lambda_2+3\theta_2) & \lambda_2 g_{2,1} &\ldots
\\
\vdots & \vdots & \ddots & \ddots & \ddots & \ddots 
\end{array}\right).
\end{equation}
The process $Y_2$ admits a stationary distribution $(\rho_{2,i})_{i\in\mathbb N}$, and obviously $\pi_{2,i}=\rho_{2,i-1}$. Writing the classical balance equations and solving the derived ODE for the generating function $\psi(z)=\sum_{n=0}^{+\infty} \rho_{2,n}z^n$, we obtain :
\begin{equation}
	\label{eq:GenFunVolumeIB}
	\psi(z) = \rho_{2,0} e^{\frac{\lambda}{\theta}\int_0^z \frac{1-G_2(u)}{1-u}\,du}
\end{equation}
where $G_2(u) = \sum_{n=0}^{\infty} g_{2,i} u^i$ is the generating function of the distribution of the sizes of limit orders submitted inside the book.

Therefore, if we specify the distributions $g_1$ and $g_2$ of the sizes of incoming limit orders respectively at the best quote and inside the book, and if subsequent computations are analytically tractable, then we can analytically derive the distribution $\pi$. We study two variants of the model 1 setting:
\begin{itemize}
	\item \textbf{Model 1a}: Model 1a assumes that all limit orders submitted at the best quote or inside the book are unit-sized, i.e. $g_{1,1}=g_{2,1}=1$ and $g_{1,n}=g_{2,n}=0$ for any $n\geq 2$.
	This assumptions is the one usually made in zero-intelligence models that look for some analytical tractability \citep[see e.g.][]{Cont2010}.
	\item \textbf{Model 1b}: Model 1b assumes that all limit orders submitted at the best quote or inside the book are geometrically distributed with parameters $0<q_1<1$ and $0<q_2<1$ respectively.
	This assumption has been used in \cite{MuniToke2014} in which it has been used to compute a general average shape of an order book.
\end{itemize}
The results in these two cases are now stated.

\begin{proposition}[Model 1a]
If all limit orders at the best quote or inside the book are unit-sized, then the stationary distribution $(\pi_j)_{j\in\mathbb N^*}$ of the volume offered at the best quote is:
\begin{align}
\pi_j = (\lambda_0+\mu_A) \sum_{i=1}^\infty h_i \sum_{k=0}^{\min(i-1,j-1)} & \binom{i-1}{k} \frac{1}{(j-1-k)!} \left(\frac{\lambda_1}{\theta_1}\right)^{j-1-k} \nonumber
\\ & \times \int_0^\infty e^{-(\lambda_0+\mu)t} e^{-k\theta_1 t} (1-e^{-\theta_1 t})^{i+j-2-2k} e^{-\frac{\lambda_1}{\theta_1}(1-e^{-\theta_1 t})} \,dt,
\label{eq:pij1a}
\end{align}
and
\begin{equation}
h_i = \frac{\lambda_0 g_{0,i}}{\mu_A+\lambda_0} + \frac{\mu_A}{\mu_A+\lambda_0} e^{-\frac{\lambda_2}{\theta_2}}\frac{(\lambda_2)^{i-1}}{(\theta_2)^{i-1} (i-1)!}.
\label{eq:hiUnit}
\end{equation}
\end{proposition}
\begin{proof}
The unit-size assumption gives $G_1(z)=z$, and inserting this in equation \eqref{eq:PDEvarphi} allows a direct solving of the latter as:
\begin{equation}
	\varphi_i(z,t) = \left[1-(1-z)e^{-\theta_1 t}\right]^i \exp\left[\frac{\lambda_1}{\theta_1}\left(z-(1-z)e^{-\theta_1 t}\right)\right],
\end{equation}
which then gives after some computations, using Leibniz differentiation formula, the transition probabilities:
\begin{equation}
r_{ij}(t) = \sum_{k=0}^{\min(i,j)} \frac{i!}{k!(i-k)!(j-k)!} \left(\frac{\lambda_1}{\theta_1}\right)^{j-k} e^{-k\theta_1 t} (1-e^{-\theta_1 t})^{i+j-2k} e^{-\frac{\lambda_1}{\theta_1}(1-e^{-\theta_1 t})}.
\end{equation}
Shifting both indices by one, multiplying by $h_i$, summing over $i$ and taking the Laplace transform yields equation \eqref{eq:pij1a}.

Furthermore, still using the unit-size assumption, the stationary distribution of a process with generator defined in equation \eqref{eq:InfGenSL} is a Poisson distribution with parameter $\frac{\lambda_2}{\theta_2}$. This readily gives equation \eqref{eq:hiUnit}.
\end{proof}

\begin{proposition}(Model 1b)
If all limit orders submitted at the best quote and inside the book are i.i.d. and geometrically-distributed with parameter $q_1$ and $q_2$ respectively, then the stationary distribution $(\pi_j)_{j\in\mathbb N^*}$ of the volume offered at the best quote is:
\begin{equation}
\begin{array}{rl}
\pi_j = & (\lambda_0+\mu_A) \Bigg[ \sum_{i=1}^\infty h_i \bigg[ 
q_1\frac{\lambda_1}{\theta_1} \sum_{k=0}^{\min(i-1,j-2)} \binom{i-1}{k} \frac{1}{(j-1-k)!} \sum_{l=0}^{j-2-k} \binom{j-1-k}{l} (-1)^l 
	\nonumber
	\\ & \times \prod_{\alpha=1}^l \left(\frac{\lambda_1}{\theta_1}-\alpha(1-q_1)\right) \prod_{\beta=1}^{j-2-k-l} \left(\frac{\lambda_1}{\theta_1}+\beta(1-q_1)\right)
	\nonumber
	\\ & \times \int_0^{+\infty} (1-e^{-\theta_1 t})^{i-k} e^{-(l+k)\theta_1 t} \left[q_1+(1-q_1)e^{-\theta_1 t}\right]^{\frac{\lambda_1}{\theta_1(1-q_1)}-l-1} e^{-(\lambda_0+\mu)t}\,dt
	\nonumber \bigg]
	\\ + & \sum_{i=j}^\infty h_i \binom{i-1}{j-1} \int_0^{+\infty} e^{-(j-1)\theta_1 t} (1-e^{-\theta_1 t})^{i-j}
	\left[q_1+(1-q_1)e^{-\theta_1 t}\right]^{\frac{\lambda_1}{\theta_1(1-q_1)}} e^{-(\lambda_0+\mu)t}\,dt
	\Bigg],
\label{eq:pij1b}
\end{array}
\end{equation}
and
\begin{equation}
h_i = \frac{\lambda_0 g_{0,i}}{\mu_A+\lambda_0} + \frac{\mu_A}{\mu_A+\lambda_0} q_2^{\frac{\lambda_2}{(1-q_2)\theta_2}} \frac{(1-q_2)^{i-1}}{(i-1)!} \frac{\Gamma(i-1+\frac{\lambda_2}{(1-q_2)\theta_2})}{\Gamma(\frac{\lambda_2}{(1-q_2)\theta_2})}.
\label{eq:hiModel1b}
\end{equation}
\end{proposition}
\begin{proof}
The assumption of a geometric distribution of the sizes of limit orders inside the book gives $G_1(z)=\frac{q_1z}{1-(1-q_1)z}$. With this definition of $G_1$, we can solve equation \eqref{eq:PDEvarphi} to obtain:
\begin{equation}
	\varphi_i(z,t) = [1-(1-z)e^{-\theta_1 t}]^i \left[\frac{q_1+(1-q_1)(1-z)e^{-\theta_1 t}}{1-(1-q_1)z}\right]^{\frac{\lambda_1}{\theta(1-q_1)}}.
\end{equation}
The Leibniz differentiation formula and some computations lead to the transient probabilities:
\begin{align}
	r_{ij}(t) = & q_1\frac{\lambda_1}{\theta_1} \sum_{k=0}^{\min(i,j-1)} \frac{i!}{k!(i-k)!(j-k)!} (1-e^{-\theta_1 t})^{i-k+1}
	\nonumber
	\\ & \times \sum_{l=0}^{j-1-k} \binom{j-1-k}{l} (-1)^l e^{-(l+k)\theta_1 t} \left[q_1+(1-q_1)e^{-\theta_1 t}\right]^{\frac{\lambda_1}{\theta_1(1-q_1)}-l-1}
	\nonumber
	\\ & \times \prod_{\alpha=1}^l \left(\frac{\lambda_1}{\theta_1}-\alpha(1-q_1)\right) \prod_{\beta=1}^l \left(\frac{\lambda_1}{\theta_1}+\beta(1-q_1)\right)
	\nonumber
	\\ & + \mathbf 1_{j\leq i} \binom{i}{j} e^{-j\theta_1 t} (1-e^{-\theta_1 t})^{i-j}
	\left[q_1+(1-q_1)e^{-\theta_1 t}\right]^{\frac{\lambda_1}{\theta_1(1-q_1)}}.
\end{align}
As in the previous case, shifting the indices by one, multiplying by $h_i$, summing over $i$ and taking the Laplace transform yields equation \eqref{eq:pij1b}.

Now, as for the stationary distribution of the volume offered at the second limit, the assumption of geometrically distributed sizes of limit orders gives $G_2(z)=\frac{q_2z}{1-(1-q_2)z}$, and equation \eqref{eq:GenFunVolumeIB} gives by derivation and after some computations:
\begin{equation}
	\forall i\in\mathbb N, \rho_{2,i} = q_2^{\frac{\lambda_2}{(1-q_2)\theta_2}} \frac{(1-q_2)^i}{i!} \frac{\Gamma(i+\frac{\lambda_2}{(1-q_2)\theta_2})}{\Gamma(\frac{\lambda_2}{(1-q_2)\theta_2})},
\end{equation}
hence the result.
Note that the distribution $(\rho_{2,i})_{i\in\mathbb N}$ is a negative binomial distribution (or Polya distribution) with non-integer size parameter $\frac{\lambda_2}{(1-q_2)\theta_2}$ and probability parameter $q_2$ \citep[see e.g.][Chap. VI.]{Feller1968}.
\end{proof}

\section{Models with both partial and aggressive market orders}
\label{sec:ModelsWithBothMarketOrders}

Type-2 models allow for both partial and aggressive market orders to be submitted, i.e. $\mu>0$ and $\mu_A>0$.
With both types of market orders, the direct approach of the previous section does not provide analytically tractable results. Therefore we propose a different strategy.
In this new setting, we can keep the analytical tractability of the model at the cost of assuming that classic orders directly affecting the best quote are unit-sized. This is the only restriction: orders submitted inside the spread or inside the book can be kept with a general distribution, and the aggressive market orders are still defined the same way, obviously, with a size equal to the volume at the best quote.
With this assumption, the best quote is a birth-and-death process, for which we can compute the Laplace transform of its transition probabilities, which can be expressed using continuous fractions.
The original result dates back to \citet{Murphy1975} but a modern derivation of the result is found in \citet{Crawford2012}.
Therefore, we are able to study the stationary distribution of the volume available at the best quote without computing the transient probabilities of the process $Y$ with infinitesimal generator \eqref{eq:InfGenFL} and $\mu>0$ (to our knowledge, such a computation is still an unresolved challenge).

Assume that both partial market orders and limit orders submitted at the best quote are unit-sized.
Then the process $Y$, which translates the evolution of the volume at the best quote (minus one) without any price movement, is a birth-and-death process with constant birth-immigration rate $\lambda_1$ and linear death-emigration rate $\mu+n\theta_1$ for any $n\geq 0$.
Let $(q_{m,n}(t)), (m,n)\in\mathbb N^2, t\in[0,\infty)$ be the transition probabilities of the process $Y$, and $(\hat q_{m,n}(s)), (m,n)\in\mathbb N^2, s\in\mathbb C$ their Laplace transform, if it exists.
Let $(B_n(s))_{n\in\mathbb N}$ the real sequence defined by the two-step recurrence:
\begin{equation}
	\left\{\begin{array}{l}
	B_0(s)=1, \;\; B_1(s)=s+\lambda_1,
	\\ B_n(s)=\left(s+\lambda_1+\mu+(n-1)\theta_1\right)B_{n-1} - \lambda_1\left(\mu+(n-1)\theta_1\right) B_{n-2}, \; n\geq 2.
	\end{array}\right.
\end{equation}

Then, adapting \citet[Theorem 1]{Crawford2012} to our special case, we have for any $(m,n)\in\mathbb N^2$ such that $m\leq n$ :
\begin{equation}
\label{eq:hatqmleqn}
\hat q_{m,n}(s)=\lambda_1^{n-m}  \frac{\hat a_1}{\hat b_1+\frac{\hat a_2}{\hat b_2+\frac{\hat a_3}{\hat b_3+\ldots}}}
\triangleq
\lambda_1^{n-m} \frac{\hat a_1(s)}{\hat b_1(s) +}\frac{\hat a_2(s)}{\hat b_2(s) +}\frac{\hat a_3(s)}{\hat b_3(s) +} \ldots
\end{equation}
(The symbol $\triangleq$ is the definition equality that introduce a simplified notation for the continuous fractions.)
In the above equation the sequences $(\hat a_i)_{i\in\mathbb N^*}$ and $(\hat b_i)_{i\in\mathbb N^*}$ are defined as follows:
\begin{equation}
\hat a_i=\begin{cases}
	B_m(s) & \text{ if } i=1,
	\\ 
	- \lambda_1\left(\mu+(n+1)\theta_1\right) B_n(s) & \text{ if } i=2,
	\\ 
	- \lambda_1\left(\mu+(n+i-1)\theta_1\right) & \text{ if } i\geq 3,
\end{cases}
\end{equation}
and
\begin{equation}
\hat b_i=\begin{cases}
	B_{n+1}(s) & \text{ if } i=1,
	\\ 
	s+\lambda_1+\mu+(n+i-1)\theta_1 & \text{ if } i\geq 2.
\end{cases}
\end{equation}
The result for any $(m,n)\in\mathbb N^2$ such that $m\geq n$ is similarly written:
\begin{equation}
\label{eq:hatqmgeqn}
\hat q_{m,n}(s)=
\left(\prod_{j=n+1}^{m}(\mu+j\theta_1)\right) \frac{\hat \alpha_1(s)}{\hat \beta_1(s) +}\frac{\hat \alpha_2(s)}{\hat \beta_2(s) +}\frac{\hat \alpha_3(s)}{\hat \beta_3(s) +} \ldots
\end{equation}
where the sequences $(\hat \alpha_i)_{i\in\mathbb N^*}$ and $(\hat \beta_i)_{i\in\mathbb N^*}$ are defined as follows:
\begin{equation}
\hat \alpha_i=\begin{cases}
	B_n(s) & \text{ if } i=1,
	\\ 
	- \lambda_1\left(\mu+(m+1)\theta_1\right) B_m(s) & \text{ if } i=2,
	\\ 
	- \lambda_1\left(\mu+(m+i-1)\theta_1\right) & \text{ if } i\geq 3,
\end{cases}
\end{equation}
and
\begin{equation}
\hat \beta_i=\begin{cases}
	B_{m+1}(s) & \text{ if } i=1,
	\\ 
	s+\lambda_1+\mu+(m+i-1)\theta_1 & \text{ if } i\geq 2.
\end{cases}
\end{equation}

Hence, the Laplace transforms $(\hat q_{m,n}(s)), (m,n)\in\mathbb N^2, s\in\mathbb C$ are numerically computable using appropriate numerical methods for the computation of continuous functions.
If we now go back to the killing and resurrection of Markov processes, we have the following result.
\begin{proposition}[Type-2 models]
If all orders submitted at the best quote are unit-sized, then the stationary distribution of the volume offered at the best quote is given by equation \eqref{eq:StatDistVolBestQuote}, i.e. $\pi_j=(\lambda_0+\mu_A) \sum_{m=1}^\infty h_m \hat q_{m-1,j-1}(\lambda_0+\mu_A)$, where the  $\hat q_{m,n}$'s are given by \eqref{eq:hatqmleqn} and \eqref{eq:hatqmgeqn}.

If limit orders submitted inside the book are assumed to be unit-sized as well (model 2a), then the probabilities $h_m$ are given by equation \eqref{eq:hiUnit}. If the sizes of these limit orders are assumed to be geometrically distributed (model 2b), then the probabilities $h_m$ are obtained with equation \eqref{eq:hiModel1b}.
\end{proposition}

\section{Empirical results}
\label{sec:EmpiricalResults}

We now provide an empirical study that allows the comparison of the models theoretically solved in the previous sections.

\subsection{Data and main parameters estimation}
\label{subsec:DataMainParamertersEstimation}

We use Thom\-son-Reu\-ters tick-by-tick data for eleven stocks traded on the Paris stock exchange, from January 4th, 2010 to February 22nd, 2010. The eleven stocks under investigation are: Air Liquide (AIRP.PA, chemicals), Alstom (ALSO.PA, transport and energy), Axa (AXAF.PA, insurance), BNP Paribas (BNPP.PA, banking), Bouygues (BOUY.PA, construction, telecom and media), Carrefour (CARR.PA, retail distribution), Danone (DANO.PA, milk and cereal products), Michelin (MICP.PA, tires manufacturing), Renault (RENA.PA, vehicles manufacturing), Sanofi (SASY.PA, healthcare), Vinci (SGEF.PA, construction and engineering).
All these stocks are included in the CAC 40 French index, i.e. they are among the largest market capitalizations and most liquid stocks on the Paris stock exchange.

For each stock, for each available trading day, we consider the data from 10:00 a.m. to 16:00 p.m., i.e. a six-hour period at the heart of the trading day. The idea is to get rid of the very busy opening and closing period where the assumption of a stationary model may be difficult to fulfill. It is well-known that even on the six-hour period considered, one does observe a seasonal activity (the U-shaped pattern of financial activity), and one might get better "stationary" results by shortening the daily period under investigation. However, it will appear that our models actually provides surprisingly satisfying results even using the full day sample.
For each stock, for each trading day, we compute the total numbers and the distributions of the sizes of: limit orders inside the spread; limit orders at the best quote; limit orders at the second best limit; partial market orders; and aggressive market orders. We also compute the time-weighted empirical distribution of the volume offered at the best quote and at the second best limit.

Straightforwardly, the estimators of the Poisson parameters $\lambda_0, \lambda_1, \lambda_2, \mu$ and $\mu_A$ (if needed by the model) are defined as the number of the associated events (aggressive limit order, limit order at the best quote, limit order inside the book, partial market orders, aggressive market orders) divided by the length of the time interval. For each of these types of orders, we also compute their respective mean order size $\sigma_0, \sigma_1, \sigma_2, \sigma_{\mu}$ and $\sigma_{\mu_A}$.
As for the cancellation parameters, we do not have any data allowing us to track the submitted orders individually, and therefore we cannot easily estimate an average lifetime $\theta_1^{-1}$ and $\theta_2^{-1}$ for cancelled orders at the best quote and inside the book.
However, we can get an order of magnitude by using equilibrium relations of incoming and outgoing flows of the order book. Our data let us compute the time-weighted average volume offered at the best quote $L_1$ and at the second best quote $L_2$. Then equating the average number of incoming and outgoing share in the order book, we set $\theta_1$ and $\theta_2$ so that $\lambda_1\sigma_1 = \mu\sigma_{\mu}+\mu_A\sigma_{\mu_A}+\theta_1 L_1$ and $\lambda_2\sigma_2 = \theta_2 L_2$.

Finally, following our theoretical framework, we will assume that all partial market orders are unit-sized, and therefore rescale all size and volume quantities by the average trade size. The rescaling gives us the empirical versions of the distributions $(g_{0,i}), (g_{1,i}), (g_{2,i})$ and $(\pi_{2,i})$ (if needed by the model) with a support roughly included in ${1,\ldots,50}$ (stock-dependent order of magnitude).

\subsection{Fitted models and benchmarks}

In the previous section we have presented two types of models (namely type-1, without partial market order, and type-2, with unit-sized partial market orders), each type having two variants (namely, a and b). In variants a, the distribution $(g_{2,i})_{i\in\mathbb N^*}$, which represents the volume of the limit orders submitted inside the book, is a Dirac distribution on the atom $1$.
In variants b, it is geometric with parameter $0<q_2\leq 1$.
The distribution $(g_{0,i})_{i\in\mathbb N^*}$, which represents the volume of the limit orders submitted inside the spread, has not been specified up to now. In line with the other assumptions, we will assume in this empirical section that this distribution is a Dirac on $1$ in variants a, and a geometric distribution with parameter $0<q_0\leq 1$ in variants b.

We now add further elements of comparison for our models.
First, for each type of model, we add a variant c in which both distributions $(g_{0,i})_{i\in\mathbb N^*}$ and $(g_{2,i})_{i\in\mathbb N^*}$ are taken equal to their empirical counterpart.
Furthermore, in order to underline the importance of the mechanism that takes into account aggressive orders and consequent upward and downward movements of the best price, we recall as benchmark the following simplistic model of the best quote, in which aggressive market and limit orders are ignored, which is equivalent to assume a constant price in our setting. This benchmark will be referred to as the Type-0 model.
In the Type-0 model, limit orders arrive at rate $\lambda_1$ with volume distribution $(g_{1,i})_{i\in\mathbb N}$ ; market orders are unit-sized and arrive at rate $\mu$ ; all standing shares have a (i.i.d.) exponential lifetime with parameter $\theta_1>0$. It is thus easily shown that the volume at the best quote in the Type-0 model is the Markov process with infinitesimal generator given at equation \eqref{eq:InfGenFL}. This process admits a stationary distribution $(\pi_i)_{i\in\mathbb N}$ satisfying the following recurrence :
\begin{equation}
	\left\{\begin{array}{rcl}
	0 & = & -\lambda_1 \pi_0 + (\mu+\theta_1)\pi_1,
	\\
	0 & = & -(\lambda_1 +\mu+n\theta_1)\pi_{n} + (\mu+ (n+1) \theta_1)\pi_{n+1} + \lambda_1 \sum_{i=1}^n g_{1,i}\pi_{n-i} \;\; (n\geq 1),
	\end{array}\right.
\end{equation}
with
\begin{equation}
	\pi_0 = \left( \frac{\mu}{\theta_1} 
	\int_0^1 u^{\frac{\mu}{\theta_1}-1} e^{\frac{\lambda_1}{\theta_1}\int_u^1 H(v)\,dv}\,du,
	\right)^{-1}
\end{equation}
where $H(u)=\frac{1-G_1(u)}{1-u}$, $G_1$ being the generating function of the distribution $(g_{1,i})_{i\in\mathbb N}$ \citep[see e.g.][section 4]{MuniToke2014}.
We will consider two cases: unit-sized limit orders (Model 0a) and limit orders with geometrically-distributed size with parameter $q_1$ (Model 0b).
In line with the general models, one may assume that the last share cannot be cancelled or executed by shifting the indices of distribution $(\pi_i)_{i\in\mathbb N}$ by $1$ (i.e. on $\mathbb N^*$).

Finally, we add a second benchmark by simulating our general zero-intelligence model of the best quote with all distributions of the model equal to their empirical counterpart. This could be considered as the "best effort" of a zero-intelligence mechanism with both classic and aggressive market and limit orders to compute the distribution of the volume offered at the best quote. Note that we cannot solve analytically the distribution in this general setting: the distribution in this case is numerically estimated by simulation. This simulated benchmark will be referred to as Type-3 model.

Table \ref{table:ModelsSummary} summarizes the models that are under study here.
\begin{table}
\begin{center}
\begin{tabular}{|c|c|c|cc|c|c|}
\hline
Model type & $(g_{0,i})_{i\in\mathbb N}$ & $(g_{1,i})_{i\in\mathbb N}$ & $(g_{2,i})_{i\in\mathbb N}$ & $(\pi_{2,i})_{i\in\mathbb N}$ & $\mu>0$ ? & $\lambda_0>0, \mu_A>0$ ?
\\ \hline
Model 0a & None & Unit-size & None & None & Yes & No
\\ Model 0b & None & Geometric & None & None & Yes & No
\\ \hline
Model 1a & Unit-size & Unit-size & Unit-Size & \textit{Poisson} & No & Yes
\\ Model 1b & Geometric & Geometric & Geometric & \textit{Neg. binomial} & No & Yes
\\ Model 1c & Empirical & Geometric & \textit{None} & Empirical & No & Yes
\\ \hline
Model 2a & Unit-size & Unit-size & Unit-Size & \textit{Poisson} & Yes & Yes
\\ Model 2b & Geometric & Unit-size & Geometric & \textit{Neg. binomial} & Yes & Yes
\\ Model 2c & Empirical & Unit-size & \textit{None} & Empirical & Yes & Yes
\\ \hline
Model 3 & Empirical & Empirical & \textit{None} & Empirical & Yes & Yes
\\ \hline
\end{tabular}
\caption{Summary of the different models and their variants. Italic means that the distribution $(\pi_{2,i})_{i\in\mathbb N}$ (resp. $(g_{2,i})_{i\in\mathbb N}$) in these variants is not a free parameter, but a consequence of the choice of $(g_{2,i})_{i\in\mathbb N}$ (resp. $(\pi_{2,i})_{i\in\mathbb N}$).}
\label{table:ModelsSummary}
\end{center}
\end{table}
On Figure \ref{fig:AIRP.PA_20100101_20100228_ISSizeLimit} the empirical distribution for the stock AIRP.PA of the size of the limit orders submitted inside the spread is plotted, along with its geometric MLE fit.
\begin{figure}
\begin{tabular}{cc}
\includegraphics[angle=270,width=0.47\textwidth]{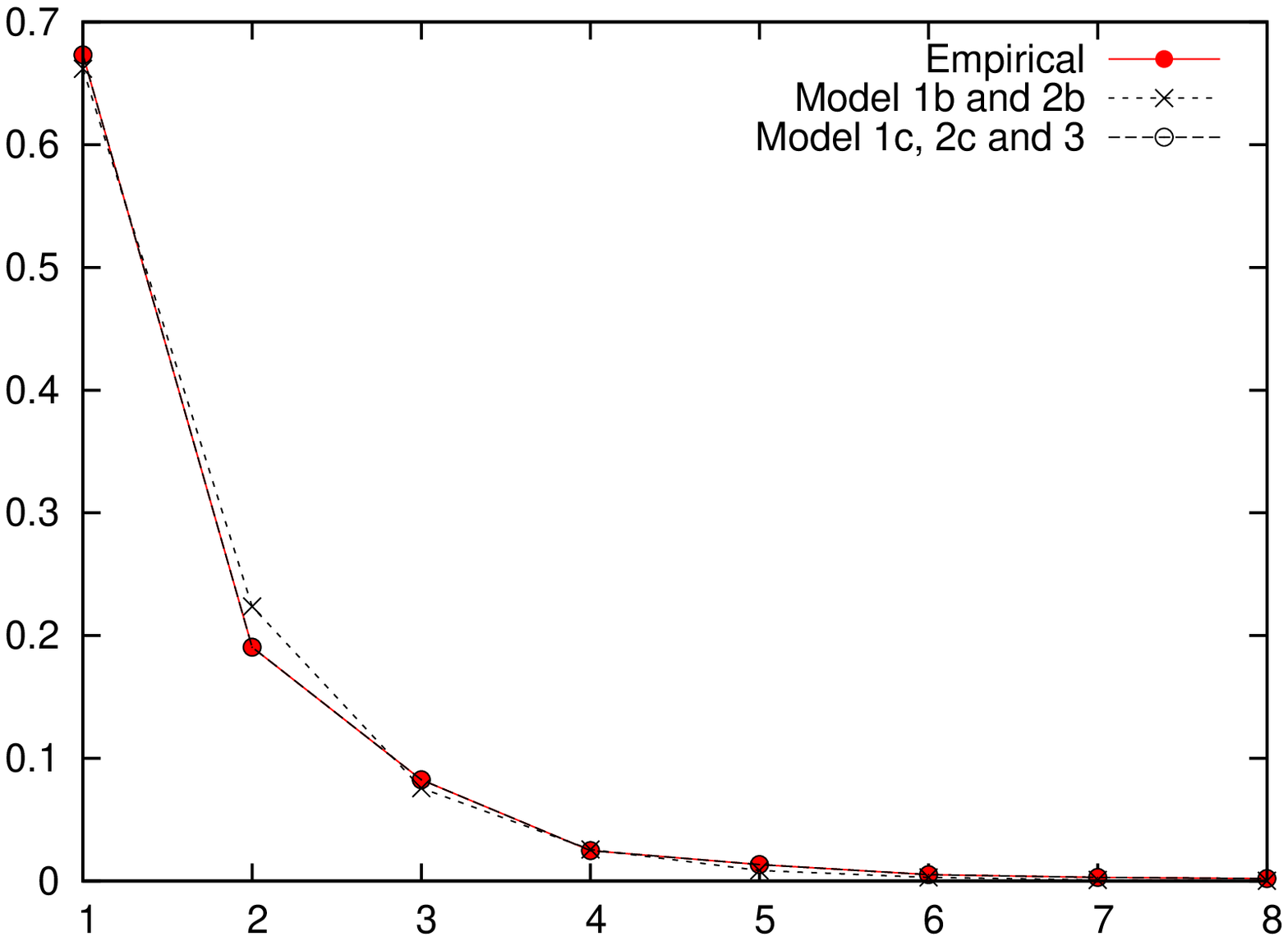}
&
\includegraphics[angle=270,width=0.47\textwidth]{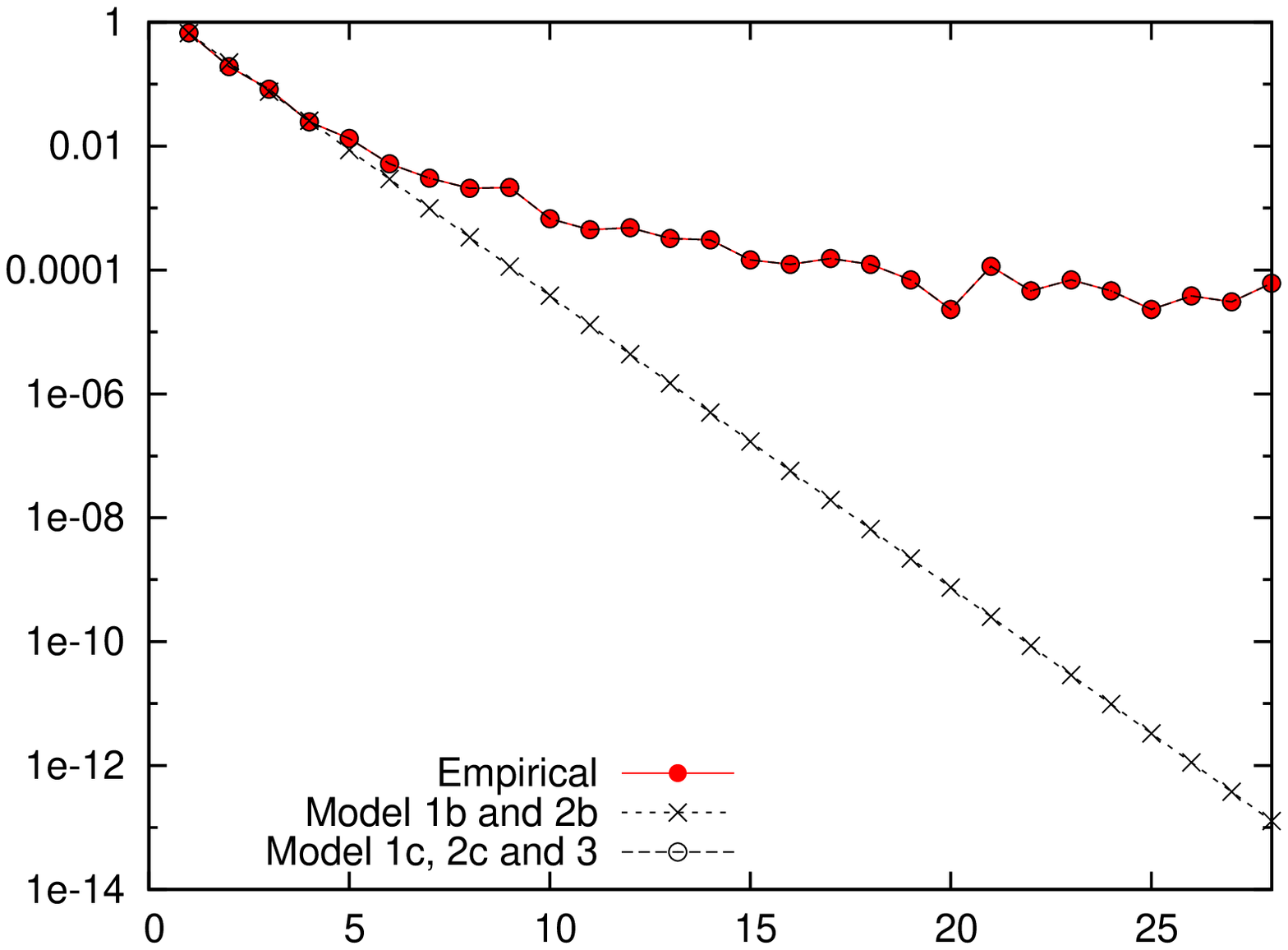}
\end{tabular}
\caption{Distribution $(g_{0,i})_{i\in\mathbb N}$ (empirical and geometric fit) of the size of limit orders submitted inside the spread, main body (left panel) and tail in semi-log scale (right panel). These distributions are computed for the stock AIRP.PA, but all stocks studied give similar results.}
\label{fig:AIRP.PA_20100101_20100228_ISSizeLimit}
\end{figure}
As expected, the geometric fit (used in model 1b and 2b) fits well the main body of the distribution, but its tail is to thin to reproduce the empirical distribution (which is used as is in Models 1c, 2c and 3).

On Figure \ref{fig:AIRP.PA_20100101_20100228_IBVolume} the empirical distribution for the stock AIRP.PA of the volume offered at the second best quote is plotted, along with its modelled counterparts.
\begin{figure}
\begin{tabular}{cc}
\includegraphics[angle=270,width=0.47\textwidth]{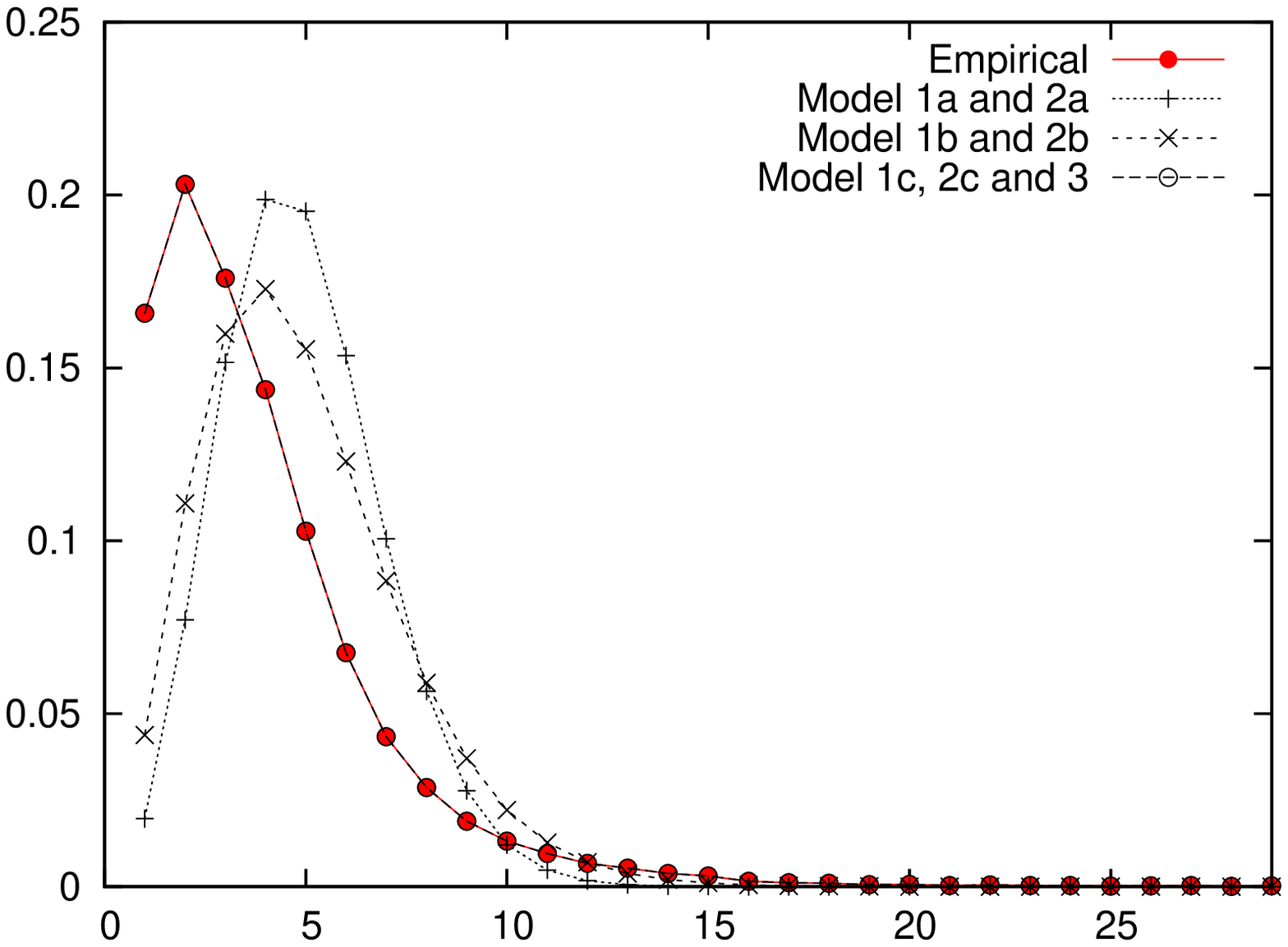}
&
\includegraphics[angle=270,width=0.47\textwidth]{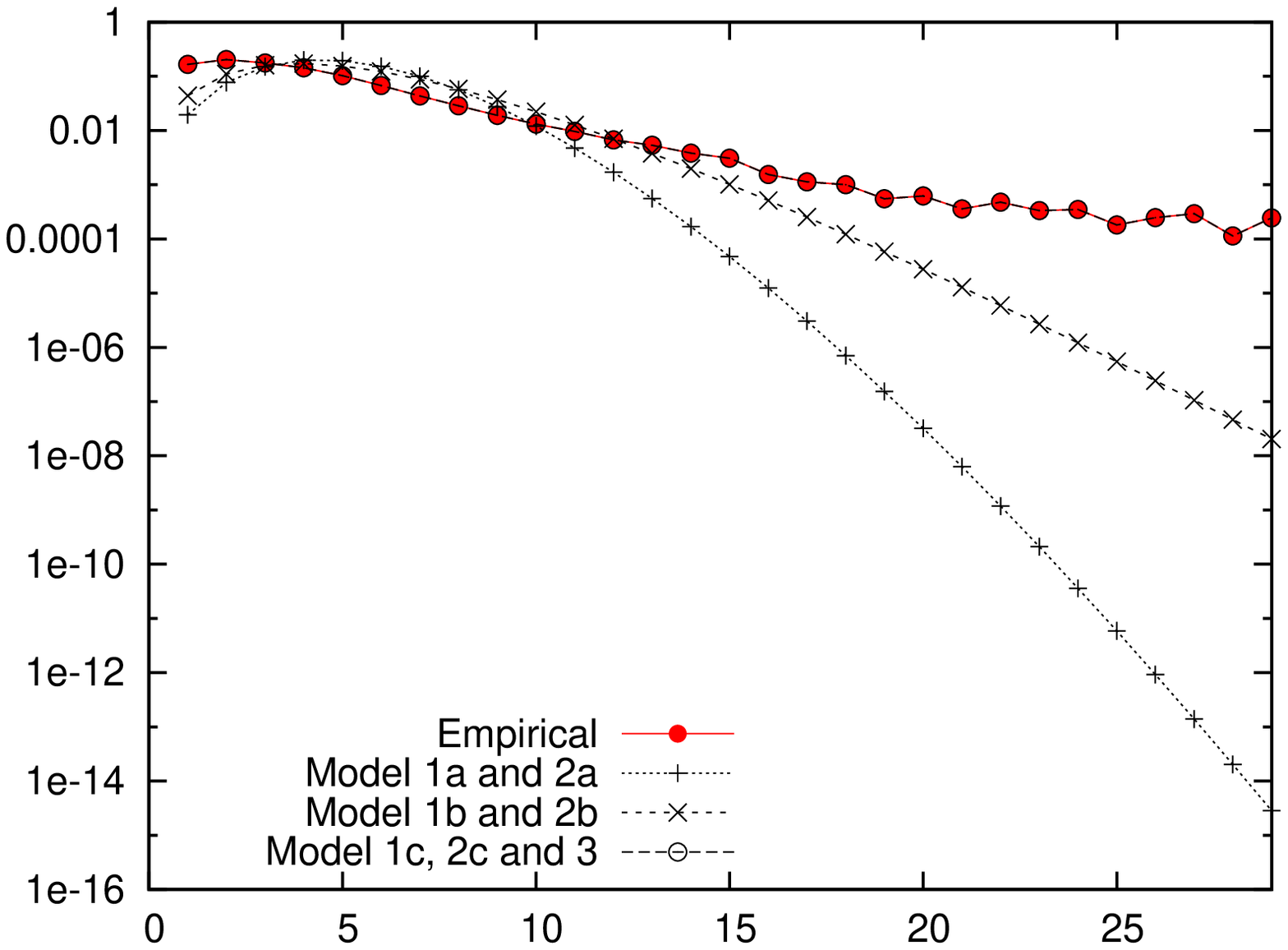}
\end{tabular}
\caption{Distribution $(\pi_{2,i})_{i\in\mathbb N}$ (empirical, Poisson obtained with models 1 and 2a and Negative binomial obtained with models 1b and 2b) of the volume offered at the second best quote, main body (left panel) and tail in semi-log scale (right panel). These distributions are computed for the stock AIRP.PA, but all stocks studied give similar results.}
\label{fig:AIRP.PA_20100101_20100228_IBVolume}
\end{figure}
As expected, the empirical distribution is very different from the one obtained using unit-size orders (Models 1a and 2a): the empirical main body is shifted to the left and its tail is fatter. Using geometrically-distributed sizes of orders (Model 1b and 2b) yields a better fit, with similar defects but of smaller amplitude. Models 1c, 2c and 3 will use the empirical distribution as input.

\subsection{Empirical results}

For all the stocks of our sample, we compute the analytical distribution of the volume offered at the best quote for all our models and variants, as well as the distribution obtained by simulation of the model 3.
On figure \ref{fig:BQVolume} we compare these analytical distributions with their empirical counterpart. For brevity only three stocks are shown, but the results for all stocks are similar.
\begin{figure}
\begin{tabular}{cc}
\includegraphics[angle=270,width=0.47\textwidth]{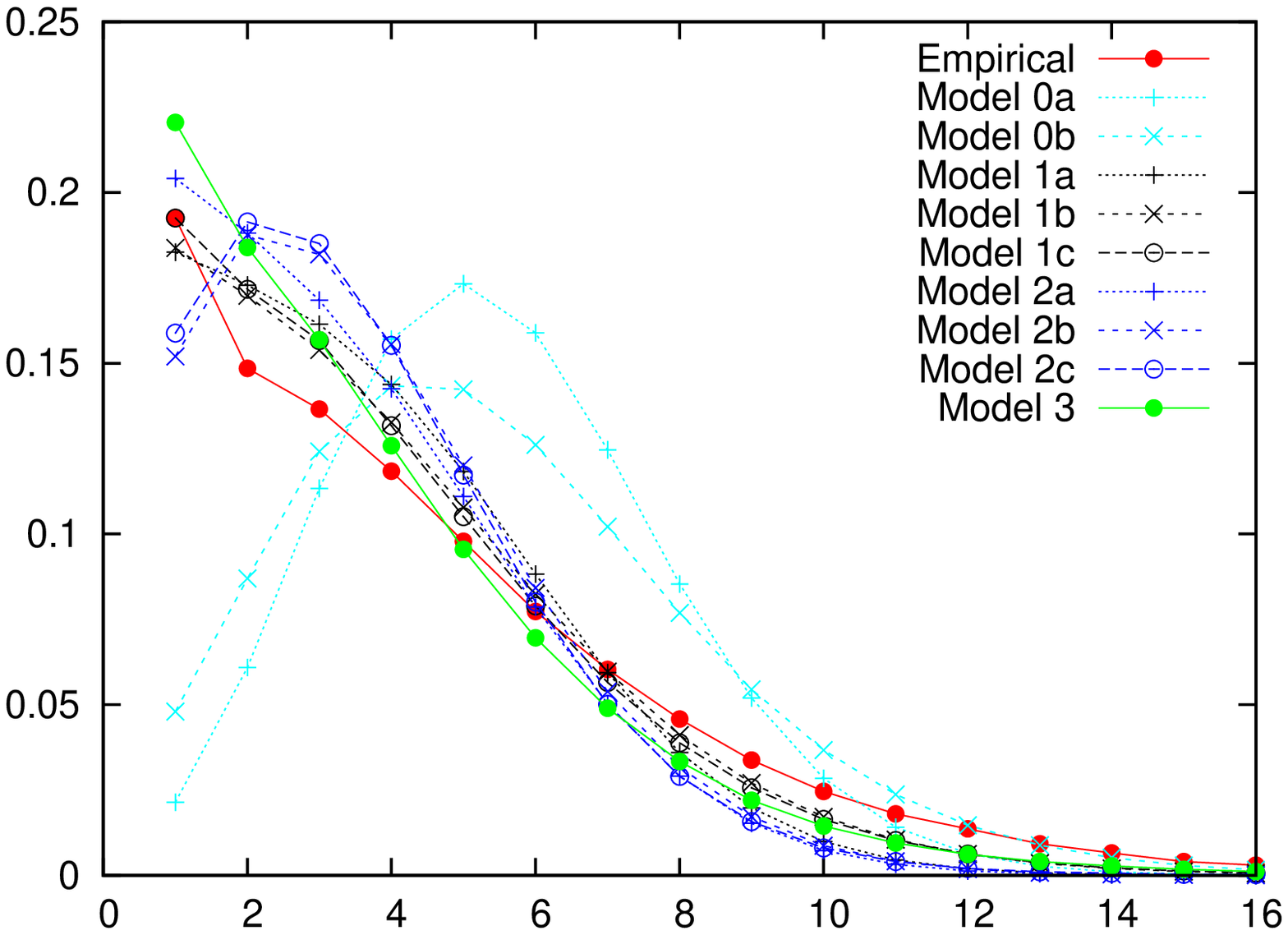}
&
\includegraphics[angle=270,width=0.47\textwidth]{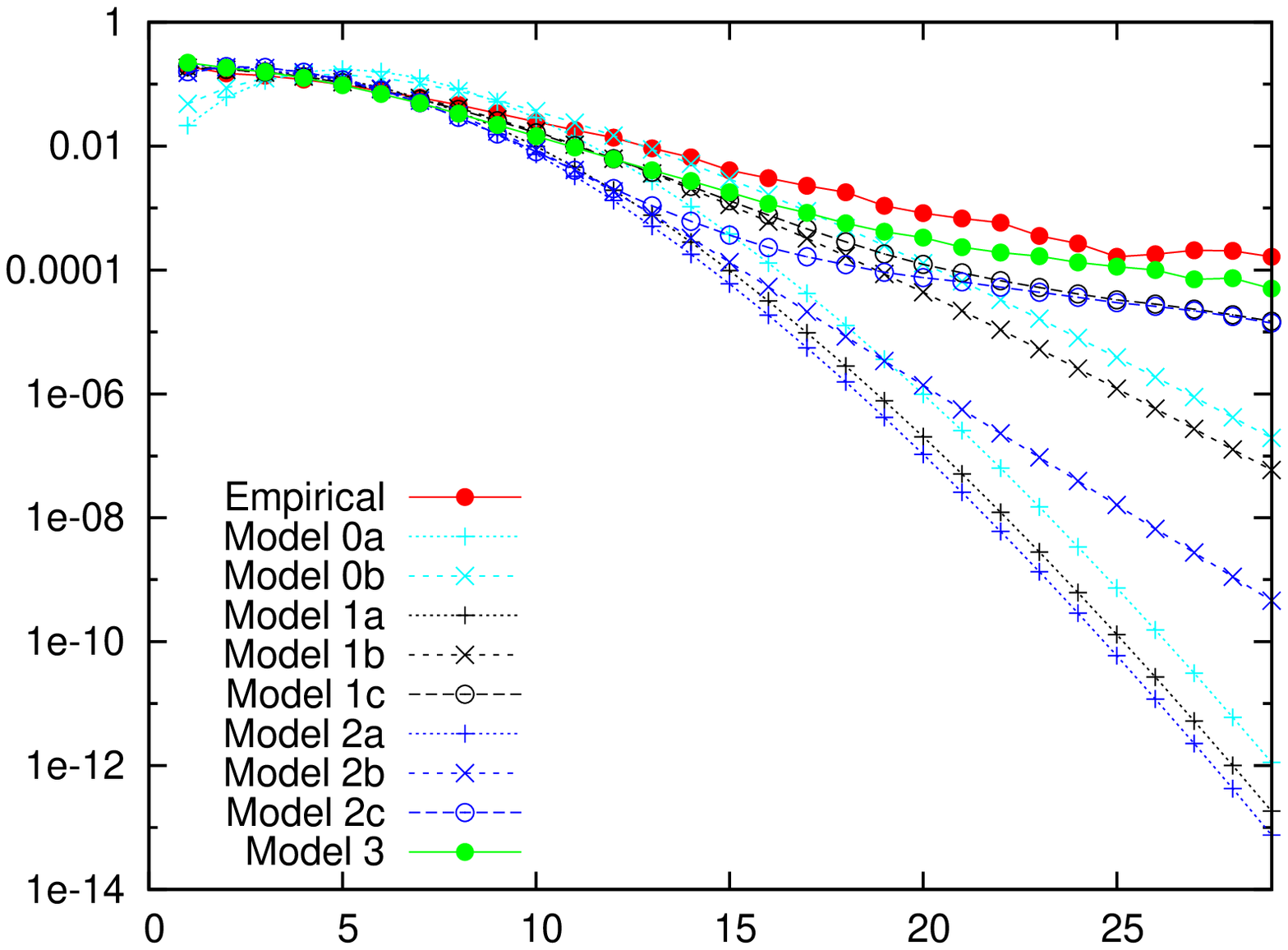}
\\
\includegraphics[angle=270,width=0.47\textwidth]{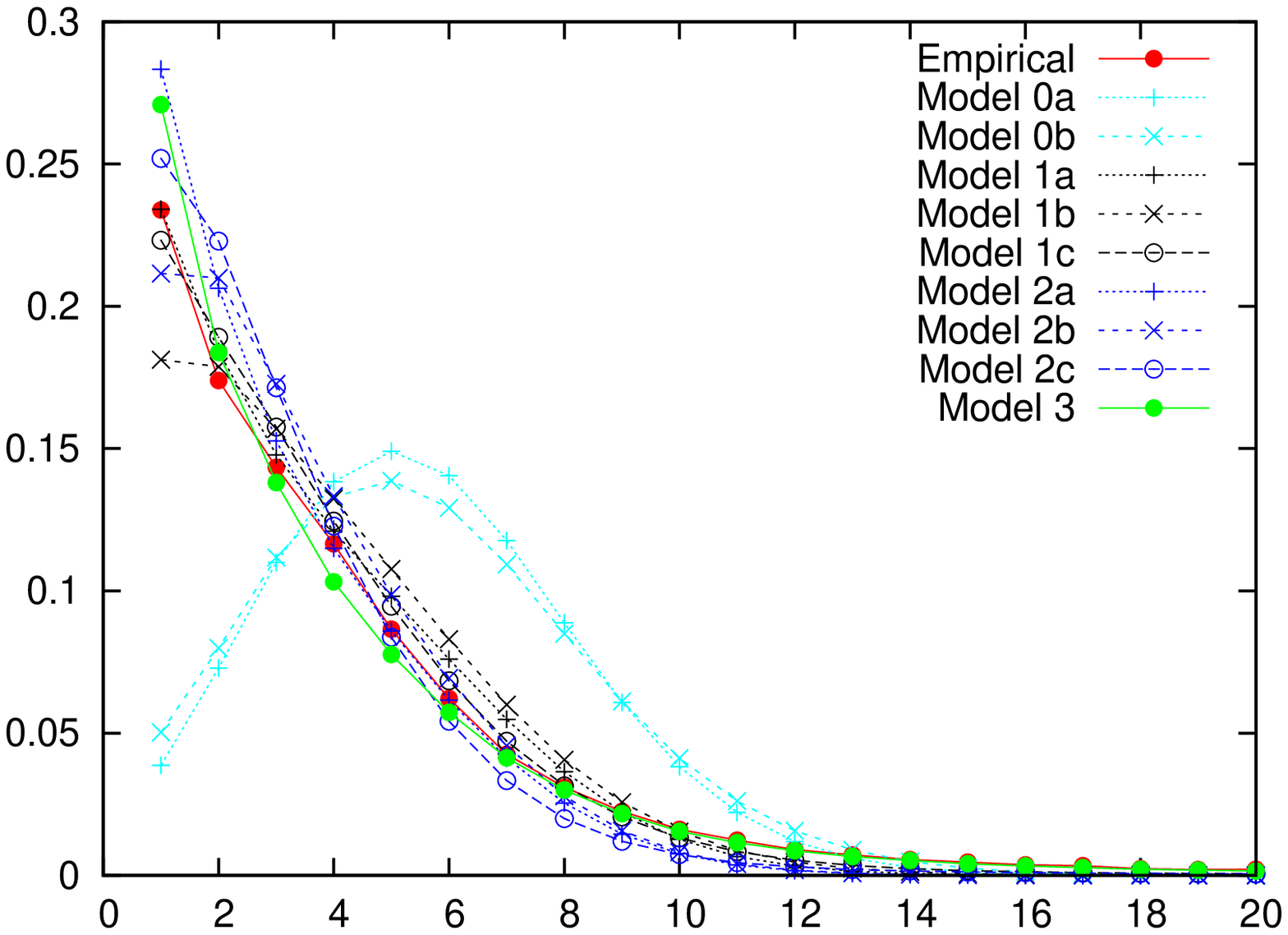}
&
\includegraphics[angle=270,width=0.47\textwidth]{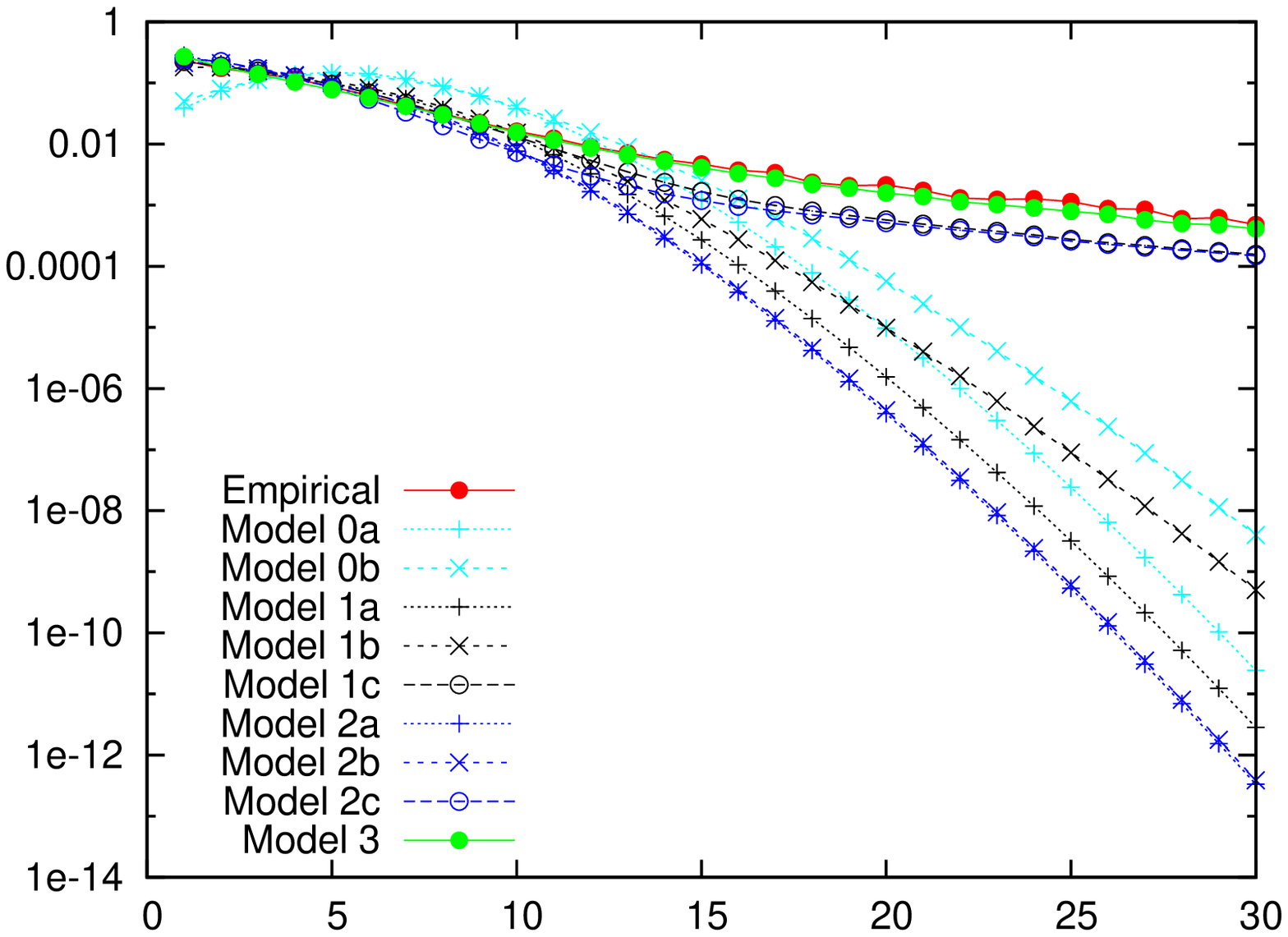}
\\
\includegraphics[angle=270,width=0.47\textwidth]{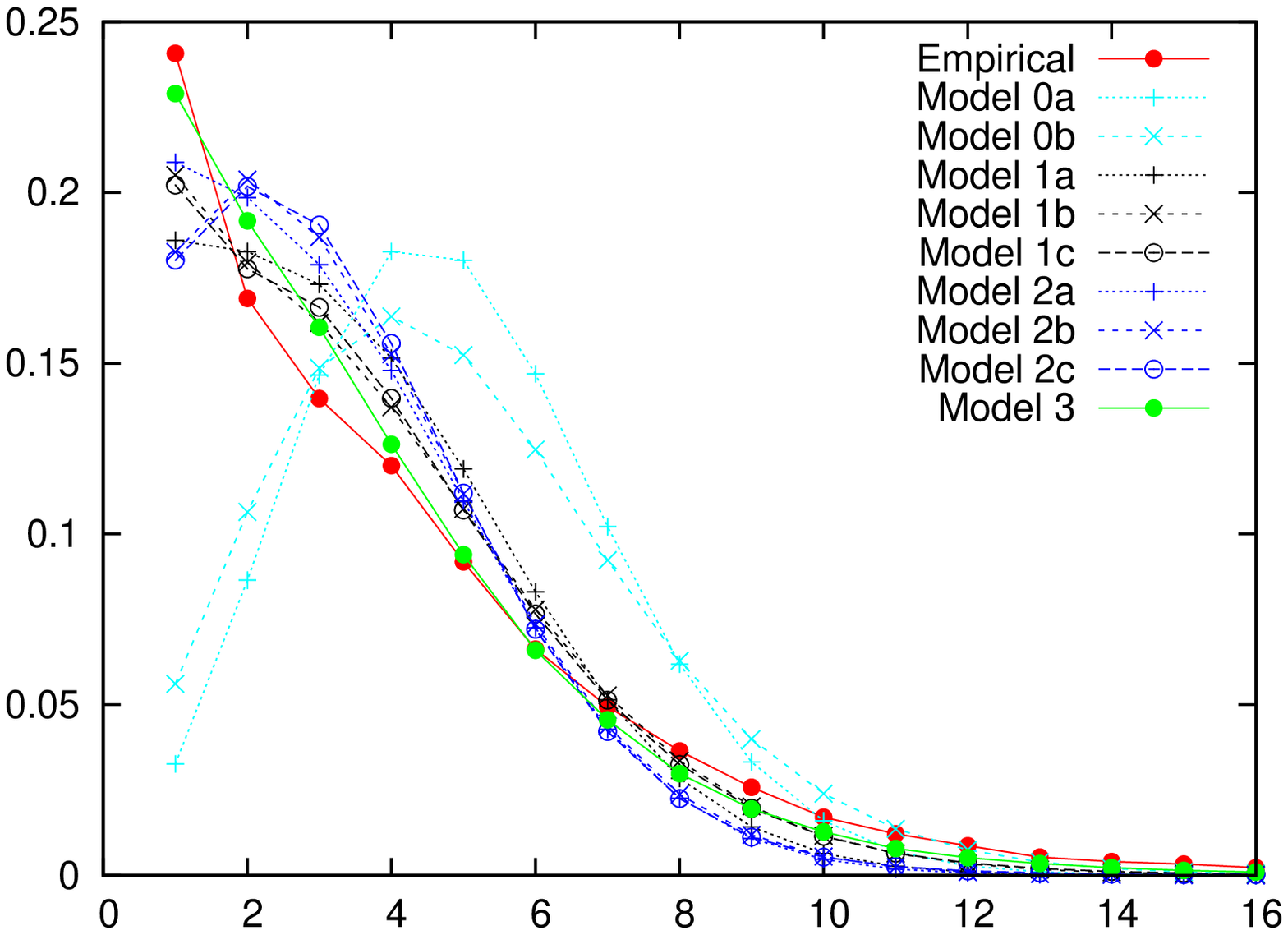}
&
\includegraphics[angle=270,width=0.47\textwidth]{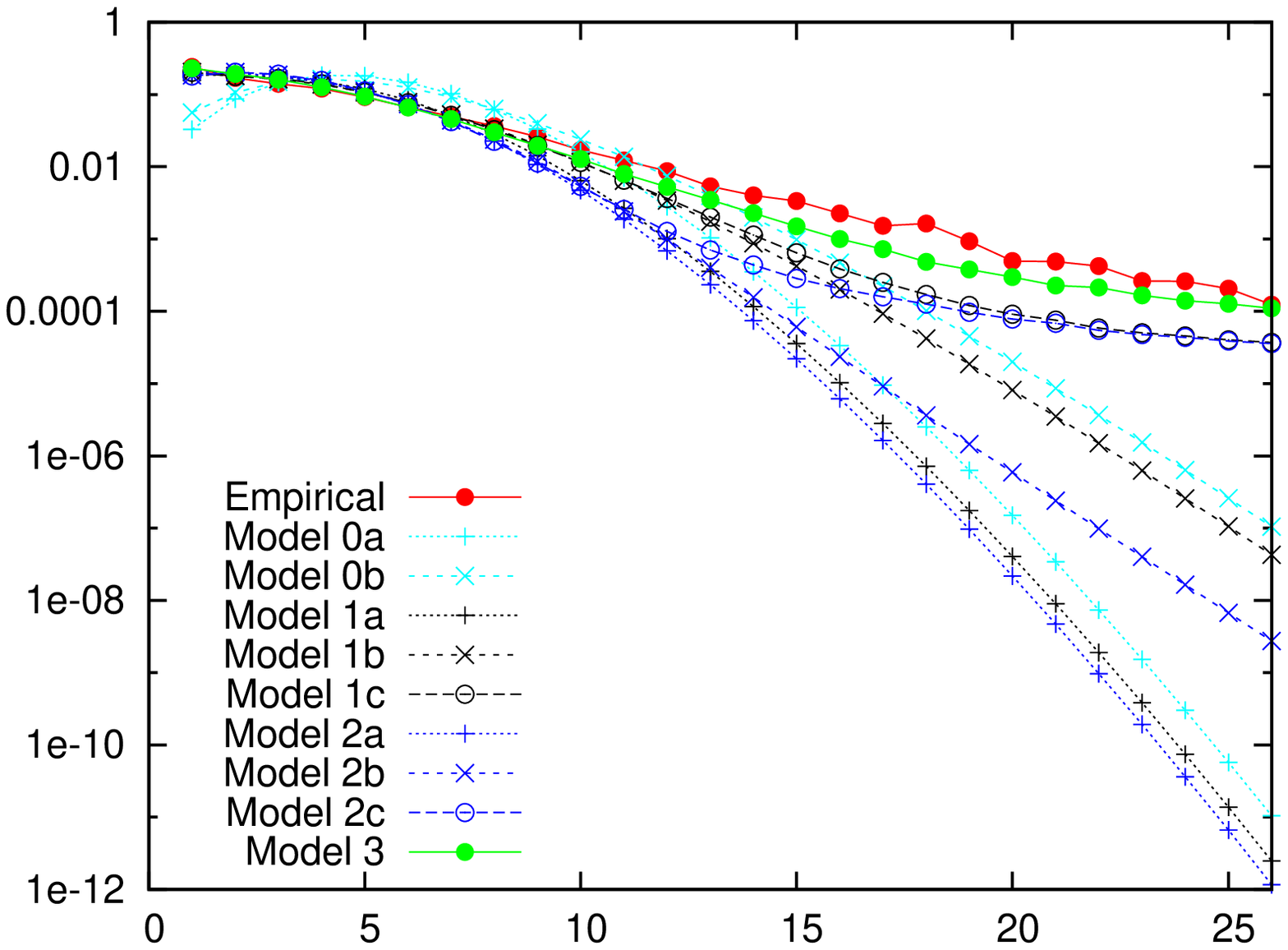}
\end{tabular}
\caption{Distribution $(\pi_{i})_{i\in\mathbb N}$ of the volume offered at the best quote for all the models described, compared to the empirical one, for three stocks: AIRP.PA (top), BNPP.PA (middle) and MICP.PA (bottom). The main body of the distribution is shown on the left column, the tail in semi-log scale is shown on the right column.}
\label{fig:BQVolume}
\end{figure}
The main result is that the killing and resurrecting mechanism used to compute the analytical distributions of models of types 1 and 2 is very efficient to produce an empirically-sound shape of the distribution of the volume offered at the best quote in an order book.
Indeed, all models of type 1 and 2 exhibit a main shape close to the empirical one, as well as to the one of model 3.
This means that the analytical formulas obtained above are a good alternative to compute a distribution that up to now was only obtained through simulation.
On the contrary, on all stocks, the basic models 0a and 0b fail to reproduce the empirical distribution.  This outlines the importance of the flows of aggressive limit orders and aggressive market orders when describing the volume at the best quote. This quantity cannot be modelled by a simple queue with limit order, partial cancellation and partial market order.

Furthermore, as expected, large differences between the models are observed when looking at the tail of the distribution. Variants a, with unit-sized orders, exhibit the thinner tails. The geometric distributions for the size of orders at the best quote and in the book (variants b) allows for slightly fatter tails, but still much thinner than the empirically observed one. Very interestingly, variants c exhibit a tail of the analytical distribution of the volume at the best quote that is as fat as the empirical one. This shows that using the empirical distributions for the distributions $(g_{0,i})_{i\in\mathbb N}$ (aggressive limit orders) and $(\pi_{2,i})_{i\in\mathbb N}$ (after an aggressive market order) is sufficient to obtain a fat tail at the best quote, even if the distribution of the size of the limit orders arriving at the best quote is thin-tailed. This again underlines the prime importance of the aggressive order flows modelled with the killing and resurrecting mechanism in our models.

We further compare the models by computing the distance (in $L^2$-norm) between the empirical distribution of the volume at the best quote and its analytical/numerical counterpart in all the models described, on all the stocks of the sample. 
Then we compute an average distance on the sample, and rank all the models from the best to the worst fit.
Table \ref{table:ModelL2Comparison} presents all the numerical results.
\begin{table}
\small
\begin{center}
\begin{tabular}{|c|cccccccl|}
\hline
Model & AIRP.PA & ALSO.PA & AXAF.PA & BNPP.PA & BOUY.PA & CARR.PA & DANO.PA & \ldots 
\\ \hline
 0a & 5.76E-002 & 6.23E-002 & 5.08E-002 & 7.10E-002 & 7.06E-002 & 6.41E-002 & 5.74E-002 & \ldots
\\ 0b & 3.32E-002 & 4.00E-002 & 4.27E-002 & 6.08E-002 & 4.45E-002 & 3.71E-002 & 3.30E-002 & \ldots
\\ 1a & 3.47E-003 & 3.12E-003 & 9.81E-004 & 8.57E-004 & 4.66E-003 & 4.65E-003 & 1.84E-003 & \ldots
\\ 1b & 1.47E-003 & 2.03E-003 & 3.28E-003 & 4.65E-003 & 2.31E-003 & 1.57E-003 & 5.98E-004 & \ldots
\\ 1c & 1.55E-003 & 2.77E-003 & 1.13E-003 & 8.33E-004 & 2.93E-003 & 9.46E-004 & 1.02E-003 & \ldots
\\ 2a & 5.02E-003 & 6.38E-003 & 5.58E-003 & 4.04E-003 & 6.42E-003 & 4.50E-003 & 4.65E-003 & \ldots
\\ 2b & 8.40E-003 & 6.37E-003 & 2.78E-003 & 3.54E-003 & 8.18E-003 & 8.56E-003 & 5.81E-003 & \ldots
\\ 2c & 8.52E-003 & 7.43E-003 & 4.88E-003 & 4.20E-003 & 8.57E-003 & 7.64E-003 & 5.41E-003 & \ldots
\\ 3d & 3.27E-003 & 1.06E-002 & 1.53E-003 & 1.79E-003 & 9.07E-003 & 4.66E-003 & 7.37E-003 & \ldots
\\ \hline
\end{tabular}
\begin{tabular}{|c|rcccc|c|c|}
\hline
Model & \ldots & MICP.PA & RENA.PA & SASY.PA & SGEF.PA & Average & Rank
\\ \hline
 0a & \ldots & 7.20E-002 & 5.56E-002 & 4.86E-002 & 6.73E-002 & 6.16E-002 & 9
\\ 0b & \ldots & 4.99E-002 & 3.42E-002 & 3.68E-002 & 3.60E-002 & 4.07E-002 & 8
\\ 1a & \ldots & 6.85E-003 & 2.28E-003 & 2.50E-003 & 4.78E-003 & 3.27E-003 & 3
\\ 1b & \ldots & 2.73E-003 & 9.57E-004 & 4.49E-003 & 1.41E-003 & 2.32E-003 & 2
\\ 1c & \ldots & 3.19E-003 & 1.42E-003 & 2.39E-003 & 1.14E-003 & 1.76E-003 & 1
\\ 2a & \ldots & 5.41E-003 & 4.80E-003 & 1.66E-003 & 5.55E-003 & 4.91E-003 & 5
\\ 2b & \ldots & 9.08E-003 & 6.06E-003 & 4.05E-003 & 1.04E-002 & 6.66E-003 & 6
\\ 2c & \ldots & 9.86E-003 & 7.13E-003 & 4.20E-003 & 9.12E-003 & 7.00E-003 & 7
\\ 3d & \ldots & 1.30E-003 & 5.46E-003 & 1.42E-004 & 6.96E-003 & 4.74E-003 & 4
\\ \hline
\end{tabular}
\caption{Distance in $L^2$-norm between the distribution of the volume at the best quote predicted by the various models and the empirical observation. Last column ranks the models from the best to the worst fit.}
\label{table:ModelL2Comparison}
\end{center}
\end{table}
Several observations can be made. 
Firstly, the quality of the fit of models of type 1 is in average better that models of type 2. Recall that models of type 1 have only aggressive market orders (no partial market orders) but allow for a flexible (geometric) size of limit orders at the best quote, while models of type require that all limit orders at the best quote are unit-sized. This observation underlines the benefit that one can get by allowing for a general size of orders in a limit order book model.
Secondly, we may also observe a bit surprisingly that models of type 1 even provide a better fit of the volume at the best quote than the model 3, which might appear counter intuitive since model 3 is the result of a simulation that take all zero-intelligence parameters available while others models are limited. 
This warns us that many mechanisms occur during the fitting of the data, that are more or less hidden but nevertheless important. In particular, the average distribution of the volume at the best quote is computed in our examples as a time-weighted quantity. The influence of the time is then integrated in the model through the way we fit the cancellation parameters as described in subsection \ref{subsec:DataMainParamertersEstimation}. Indeed, the cancellation parameters are not directly fitted but are a sort of adjustment variables. It turns out that models of type 1 do a better job at integrating this dimension.
As an illustration of this phenomenon, we redo the whole empirical part without taking into account the time while computing the average volume offered at the best quote. This means we work in "event time", which in other words is equivalent to considering that the timestamps of the data are deleted and randomly regenerated in the same order with exponential inter arrival times. In this setting, model 3 shows as expected the best performance: its average $L^2$-distance to the empirical distribution is equal to $2.16\times 10^{-3}$, while models of types 1 and 2 follow at a distance of $4.27\times 10^{-3}$ and more.

\section{Conclusion}
This paper has shown that basic zero-intelligence models of the limit order book are able to accurately describe the stationary distribution of volume offered at the best quote providing they include a proper mechanism to take into account aggressive orders that move the price: aggressive limit orders submitted inside the spread, as well as aggressive market orders that remove the whole liquidity available at the first level of the book.
Here, we have modelled these aggressive orders using results on killing and resurrecting a Markov process (here, the evolution of the quantity at the best quote when the price is constant), which allows us to provide analytical formulas for the distribution of interest.

We end this section by one last observation that may trigger future work on this subject. Figure \ref{fig:AIRP.PA_20100101_20100228_IBVolume} show that the theoretical models for $(\pi_{2,i})_{i\in\mathbb N}$ have the same defects in modelling the empirical distribution of the volume inside the book as Type-0 models with respect to the volume at the best quote: namely, a body shifted to the right and a tail too thin. This is not surprising as the volume inside the book as been treated here as one entity, without taking into account killing and resurrecting (as in models of Type-0).
Future work may include the killing and resurrecting as a cascading effect: when the price move, \emph{all} the limits in the book are shifted, i.e. are killed and resurrected.

\bibliographystyle{authordate1}
\bibliography{FirstLimitProcess} 

\end{document}